\documentclass[a4paper,man,natbib, 11pt]{article}
\RequirePackage{amsthm,amsmath,amsfonts,amssymb}
\RequirePackage{natbib}
\RequirePackage[colorlinks,citecolor=blue,urlcolor=blue]{hyperref}
\RequirePackage{graphicx}
\usepackage[margin=1.0in]{geometry}
\theoremstyle{plain}
\theoremstyle{remark}
\newtheorem{theorem}{Theorem}
\newtheorem{remark}{Remark}

\usepackage{bm}
\usepackage{float}
\usepackage{comment}
\usepackage{subcaption}
\usepackage{booktabs}
\usepackage{quotes}
\usepackage{stfloats}
\newcommand{\bA}{\mathbf{A}}
\newcommand{\bY}{\mathbf{Y}}
\newcommand{\bX}{\mathbf{X}}
\newcommand{\bV}{\mathbf{V}}
\newcommand{\bI}{\mathbf{I}}
\newcommand{\bZ}{\mathbf{Z}}
\newcommand{\bW}{\mathbf{W}}
\newcommand{\bD}{\mathbf{D}}
\newcommand{\bU}{\mathbf{U}}
\newcommand{\bu}{\mathbf{u}}

\newcommand{\bepsilon}{\boldsymbol\epsilon}
\newcommand{\T}{\intercal}

\title{Disentangling network dependence among multiple variables}
\author{Zhejia Dong\thanks{Department of Biostatistics, Brown University}, Corwin Zigler\thanks{Department of Biostatistics, Brown University}, and Youjin Lee\thanks{Department of Biostatistics, Brown University. Email: youjin\_lee@brown.edu}}
\date{}

\begin{document}
\sloppy
\maketitle
\linespread{1.15}

\begin{abstract}
When two variables depend on the same or similar underlying network, their shared network dependence structure can lead to spurious associations. 
While statistical associations between two variables sampled from interconnected subjects are a common inferential goal across various fields, little research has focused on how to disentangle shared dependence for valid statistical inference. We revisit two different approaches from distinct fields that may address shared network dependence: the pre-whitening approach, commonly used in time series analysis to remove the shared temporal dependence, and the network autocorrelation model, widely used in network analysis often to examine or account for autocorrelation of the outcome variable. We demonstrate how each approach implicitly entails assumptions about how a variable of interest propagates among nodes via network ties given the network structure. 
We further propose adaptations of existing pre-whitening methods to the network setting by explicitly reflecting underlying assumptions about "level of interaction" that induce network dependence, while accounting for its unique complexities. Our simulation studies demonstrate the effectiveness of the two approaches in reducing spurious associations due to shared network dependence when their respective assumptions hold. However, the results also show the sensitivity to assumption violations, underscoring the importance of correctly specifying the shared dependence structure based on available network information and prior knowledge about the interactions driving dependence.
\end{abstract}

\noindent%

\section{Introduction}

\textit{Network dependence}, or network autocorrelation---statistical dependence within a variable due to network ties---is prevalent across various studies where subjects are interconnected through network relationships~\citep{dow1982network, dow1984galton, o2008analysis, peeters2009network, lee2021network}. Numerous statistical methods have been developed to either account for or assess autocorrelation within a single network-dependent variable, typically the primary outcome of interest~\citep{leenders2002modeling, christakis2004social}, often focusing on its relationship with a fixed covariate.  
However, in practice, it is common for both outcome and exposure variables to exhibit network dependence.    
Despite the widespread use of methods for correcting dependence in multiple variables in other settings, 
there has been little discussion of statistical inference when examining relationships among \textit{multiple} network-dependent variables. 

\subsection{Spurious associations due to network dependence}

It is well known that two independently generated time series can appear statistically correlated if they exhibit similar patterns of strong autocorrelation. This is also known as nonsense or volatile correlations in econometrics literature~\citep{yule1926we, ernst2017yule}. In fact, any form of autocorrelation shared by two or more variables can lead to such misleading associations. \cite{lee2021network} demonstrate that two variables sampled from the same network, exhibiting shared dependence structure, may appear statistically correlated, even in the absence of a true association. Following \cite{lee2021network}, we refer to this phenomenon as \textit{spurious associations due to network dependence}. Despite the shared underlying mechanisms behind spurious associations across different fields, statistical approaches to addressing this issue remain discipline-specific. To our knowledge, no existing literature comprehensively examines the applicability of existing methods for disentangling shared dependence in the context of network dependence.

While spurious associations due to network dependence have been routinely unacknowledged in many human subjects studies, accounting for network dependence in statistical models has been widely practiced in network research, particularly when network data are explicitly collected. Such models primarily focus on assessing the presence and extent of autocorrelation~\citep{wang2014statistical,dittrich2020network} or evaluating whether network relationships causally induce correlations among observations~\citep{carr2015network,de2016social}. These analyses do not necessarily require multiple variables beyond the primary outcome of interest, and even when additional variables are included, no further statistical assumptions about them are typically required.
In this work, we focus on cases where the relationships \textit{between} two variables sampled from the same network are of interest, assuming that the presence of autocorrelation in these variables is already known and thus not our focus.

\subsection{Our contribution}

In this work, we examine the utility of two widely used approaches for accounting for autocorrelation in addressing spurious associations due to network dependence. The first approach is pre-whitening methods, which are commonly applied in time series analysis but have seen limited use in network settings, and 
the second approach is the network autocorrelation model, which is not explicitly designed to correct for shared dependence among multiple variables, but rather to address autocorrelation in the outcome variation within regression models. While neither approach was originally developed for the present setting, we demonstrate that either of them can be effective in addressing shared network dependence across multiple variables---though only under specific circumstances related to the presumed network structure.

Our main contributions are as follows. 
First, we identify key differences between network dependence and dependence arising in other contexts, and we highlight important elements that should be considered when modeling network dependence.
Second, we discuss how pre-whitening approaches can be adapted in network settings. We propose a general approach to constructing the dependence structure based on certain statistical assumptions and known network structure and introduce an estimation method that addresses computational challenges due to the complexity of network dependence, illustrated through a concrete example.
Lastly, we examine the assumptions underlying network autocorrelation models and compare them with those used in the pre-whitening approach.
Overall, our work highlights the importance of considering the level of interactions on a network that induces autocorrelation across multiple variables. This is because different approaches for addressing within-variable dependence implicitly assume different levels of interaction, which has important implications for both model specification and interpretation.  

The remainder of this paper is organized as follows. In Section~\ref{sec:pre}, we review autocorrelation in other settings and discuss the unique challenges in network dependence. Section~\ref{sec-premix} introduces the pre-whitening approach and examines its potential application to network dependence. Section~\ref{sec-nam} discusses the network autocorrelation model and examines its underlying assumptions regarding shared dependence.
Section~\ref{sec-sim} presents a simulation study comparing the performance of these two approaches in correcting spurious associations due to network dependence. Finally, we discuss the results and their implications in Section~\ref{sec-dis}.

\section{Preliminaries}\label{sec:pre}

\subsection{Structured autocorrelation in other contexts}

Autocorrelation is common in temporal, spatial, genetic, and network spaces, where the degree of dependence between two observations depends on their respective distance measures. 
In temporal autocorrelation, distance is often defined by time lags on $\mathbb{R}$ between observations, meaning that observations closer in time are more likely to be correlated than those farther apart~\citep{makridakis1994time}. 
In spatial autocorrelation, the Euclidean distance on $\mathbb{R}^{2}$ is often used to quantify dependence between observations~\citep{legendre1993spatial}. In both temporal and spatial autocorrelation, as the number of observations ($n$) increases, the distribution of temporal and spatial distances becomes more variable, leading to an exponential growth in the number of observation pairs with large (relative) distances. As a result, autocorrelation among most pairs decays. This scarcity of strongly autocorrelated pairs plays a crucial role in statistical inference by ensuring statistical independence among most observation pairs or even generating clusters across which no statistical correlations are assumed. 

Genetic autocorrelation, or genetic dependence, is one of the most commonly considered forms of autocorrelation in human subjects research. 
In genome-wide association studies (GWAS), statistical associations between diseases and genetic variants are often of interest, and each is likely to be genetically autocorrelated. For example, parents and their children may have similar disease statuses (not solely due to genetics but also because they share health-related behaviors or environmental factors), and their genetic variants are inherently correlated. Failure to adjust for shared genetic autocorrelation between these two variables can lead to inflated positives in association measures. This is also known as cryptic relatedness or population structure~\citep{astle2009population,sillanpaa2011overview}.
To quantify genetic similarity of relatedness, the kinship matrix measures pairwise genetic relationships within a population~\citep{kang2008efficient,astle2009population,speed2015relatedness,hoffman2013correcting,li2014enrichment,yu2006unified}. Higher values in the kinship matrix indicate closer genetic relationships, while lower values indicate less genetic relatedness. When distinct clusters are not present in the population, the kinship matrix is used to account for shared genetic dependencies among individuals~\citep{kang2008efficient,yu2006unified}. 
This approach assumes that the kinship matrix, often estimated using sample correlations among all single nucleotide polymorphisms (SNPs), directly captures the dependence structure that governs the autocorrelation of each variable (e.g., disease status and genetic variants)~\citep{milligan2003maximum, hayes2009increased, jiang2022unbiased}.

\subsection{Network dependence and its unique challenges}\label{ssec:challenges}

In the contexts above, autocorrelation between two units is largely determined by their distance within each context: the farther apart they are, the less correlated they tend to be. Therefore, once distance is defined or estimated, the degree of autocorrelation can be inferred. In networks, however, while distance within the network space contributes to autocorrelation, there is often no clear consensus on how to define this distance. More importantly, distance alone is not the primary factor governing the degree of autocorrelation in network settings.

Let us first focus on measuring distance in a network. Similar to the kinship matrix, distance on a network space can be measured using an adjacency matrix, denoted by $\mathbf{A} = [A_{ij}]$, where  $A_{ij} = 1$ indicates a tie (or edge) between nodes $i$ and $j$, and $A_{ij} = 0$ otherwise, for $i,j=1,2,\ldots, n$. 
One straightforward way to determine the dependence structure (e.g., variance-covariance matrix) is to define it as $\text{Corr}(Y_{i}, Y_{j}) = \sigma A_{ij}$ when $\bA$ is symmetric, or $\text{Corr}(Y_{i}, Y_{j}) = \sigma(A_{ij}+A_{ji})/2$ when $\bA$ is asymmetric, where $0 \leq \sigma \leq 1$. This assumes the presence of dependence if and only if nodes $i$ and $j$ are adjacent (e.g., being friends with each other) in a network, with the strength of dependence potentially reflecting the reciprocity of their relationships. This approach is similar to adjustments using the kinship matrix in GWAS.

Another way to define the distance between two nodes is by the geodesic distance, which counts the number of ties along the shortest paths between them from $\mathbf{A}$, and to determine the degree of dependence as the inverse of their geodesic distance (e.g., \cite{macdonald2011school,agneessens2017geodesic}). Under this definition, dependence exists for all adjacent pairs at the same amount, and as long as two nodes are connected, they are not statistically independent. This can result in a non-sparse variance-covariance structure, which, unlike temporal and spatial dependence structures---where distance is typically measured in Euclidean space---is not easily resolved by increasing the number of observations in network data. As researchers include a larger number of observations, pairs of nodes are more likely to be connected, and their distances in network space become shorter. This non-Euclidean nature renders network dependence more complicated, and it becomes challenging to guarantee asymptotic properties of many statistics developed under certain (conditional) independent structures. As a result, methods or models that are developed to accommodate temporal or spatial dependence may not be appropriate for, or at least cannot be adapted directly to network dependence. 

Dependence induced by a network is often more complex than simply the strength of network ties dictated by $\bA$: unlike other dependence settings, where the degree of dependence tends to decay smoothly with increasing distance, network dependence is arbitrary and context-specific. For example, compared to adjacent pairs $i$ and $j$ with $A_{ij} =1$, a pair of $i$ and $k$ with $A_{ik} = 0$ may exhibit higher autocorrelation in certain variables if nodes $i$ and $k$ share a larger number of mutual friends than nodes $i$ and $j$ do. For this reason, unlike genetic dependence, whose structure is directly determined by the kinship matrix, network dependence is often not solely informed by an adjacency matrix but is also heavily influenced by the underlying \textit{transmission process} through which a variable of interest propagates among nodes via network ties.

\subsection{Transmission process through network ties}

In addition to distance derived directly from $\mathbf{A}$---which is obtained independently of context related to the variables of interest---network dependence in each variable is influenced by the transmission or contagion process, which is heavily dependent on the type of variable. This process refers to the mechanism by which a variable of interest (e.g., opinions, behaviors, or attributes of nodes) propagates through network ties, thus generating autocorrelation among observations.  
Even when $\mathbf{A}$ is correctly specified (i.e., context-free network distance is accurate), misspecifying the pattern through which variables spread across the network can lead to invalid statistical or causal inferences~\citep{hays2010spatial, an2018causal,  lee2023finding}. {In this paper, we focus on the transmission processes \textit{within} a variable to highlight the specific assumptions about the transmission process on a network space that underlie existing approaches for addressing autocorrelation.} For simplicity, we illustrate the transmission process in $\mathbf{Y}$, assuming the same process applies to other variables (e.g., $\mathbf{X}$) of interest. We also consider the case where the transmission process within one variable is not affected by others, and any associations \textit{between} the variables are the focus of our inference and should not be adjusted. 

\begin{figure*}[!ht]
\centering
    \begin{subfigure}[b]{0.32\textwidth}
     \centering
     \includegraphics[width=\textwidth]{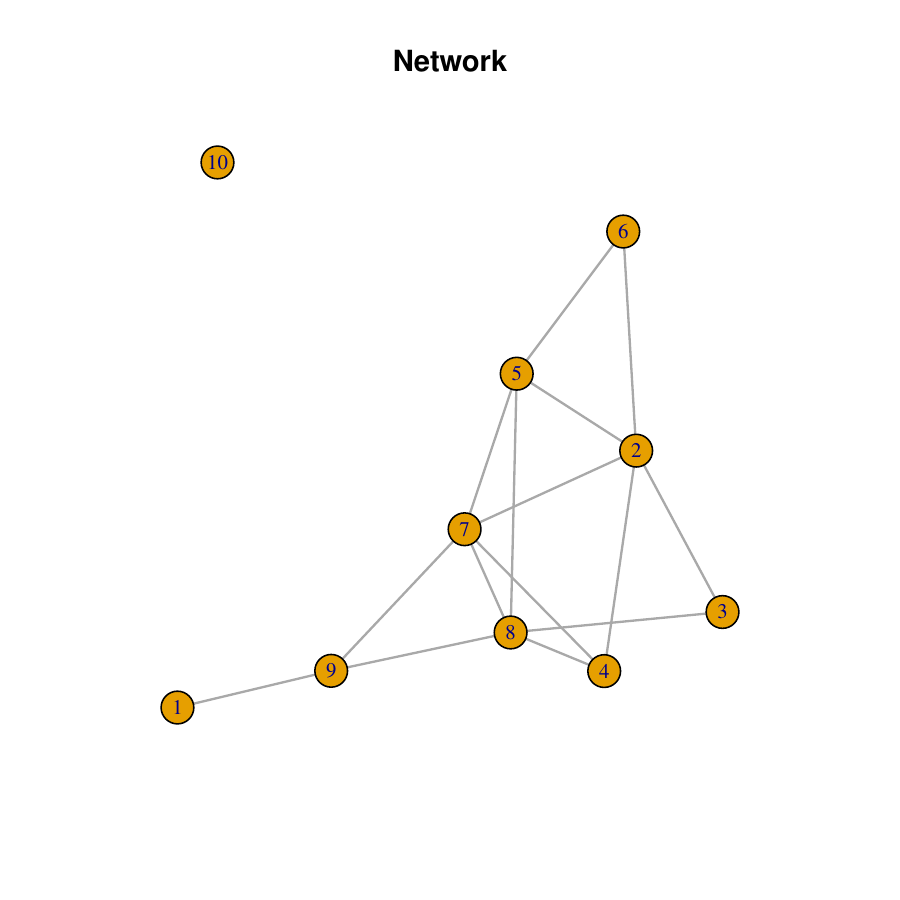} \vspace{-1.5cm}
     \caption{} 
      \label{fg:network}
    \end{subfigure}
    \begin{subfigure}[b]{0.32\textwidth}
     \centering
     \includegraphics[width=\textwidth]{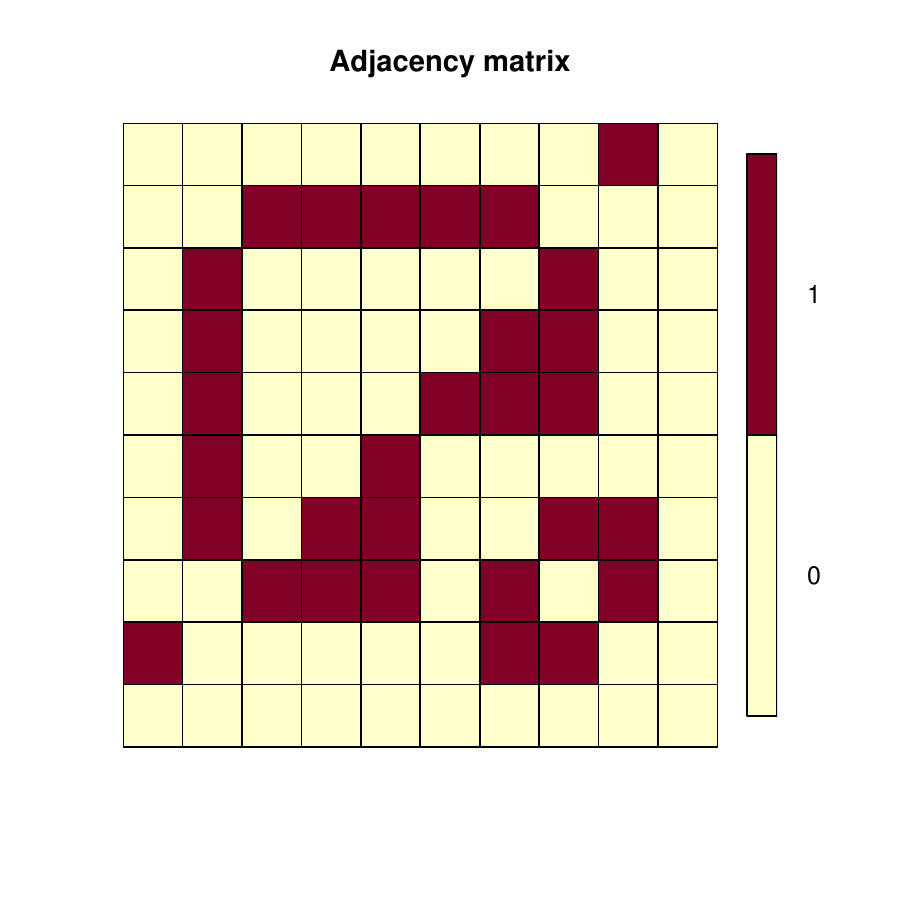} \vspace{-1.5cm}
     \caption{}  
     \label{fg:adj}
    \end{subfigure}
     \begin{subfigure}[b]{0.32\textwidth}
     \centering
     \includegraphics[width=\textwidth]{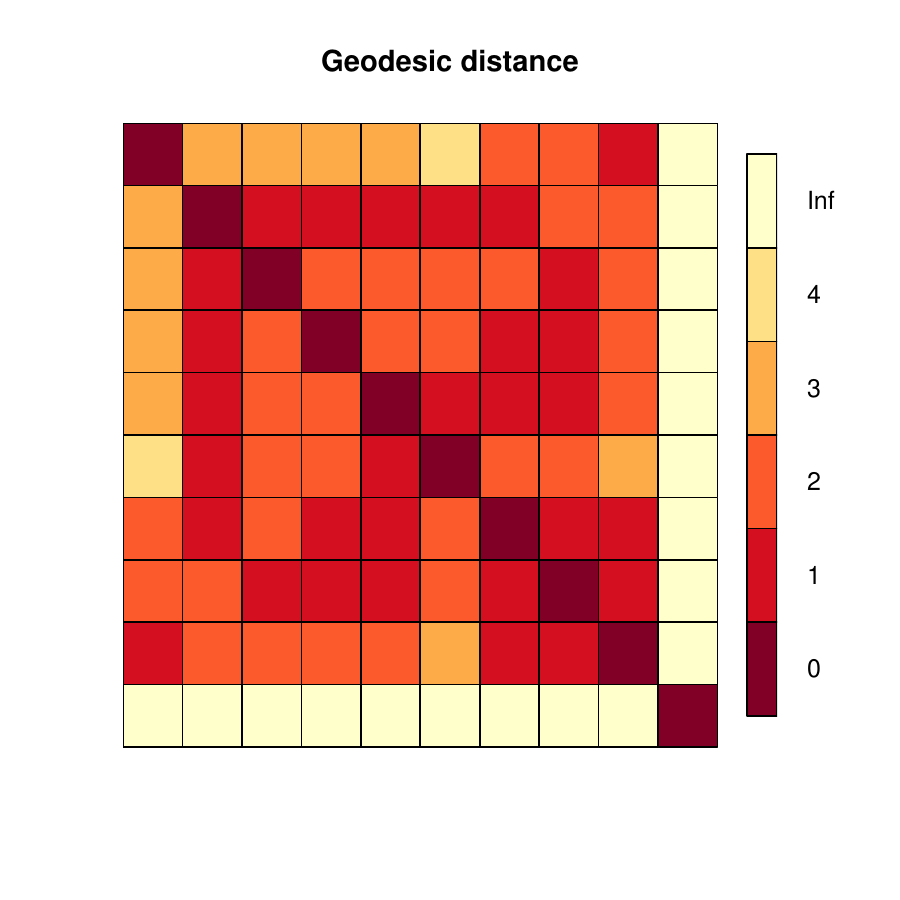} \vspace{-1.5cm}
     \caption{} 
     \label{fg:geodesic}
    \end{subfigure}
     \begin{subfigure}[b]{0.32\textwidth}
     \centering
     \includegraphics[width=\textwidth]{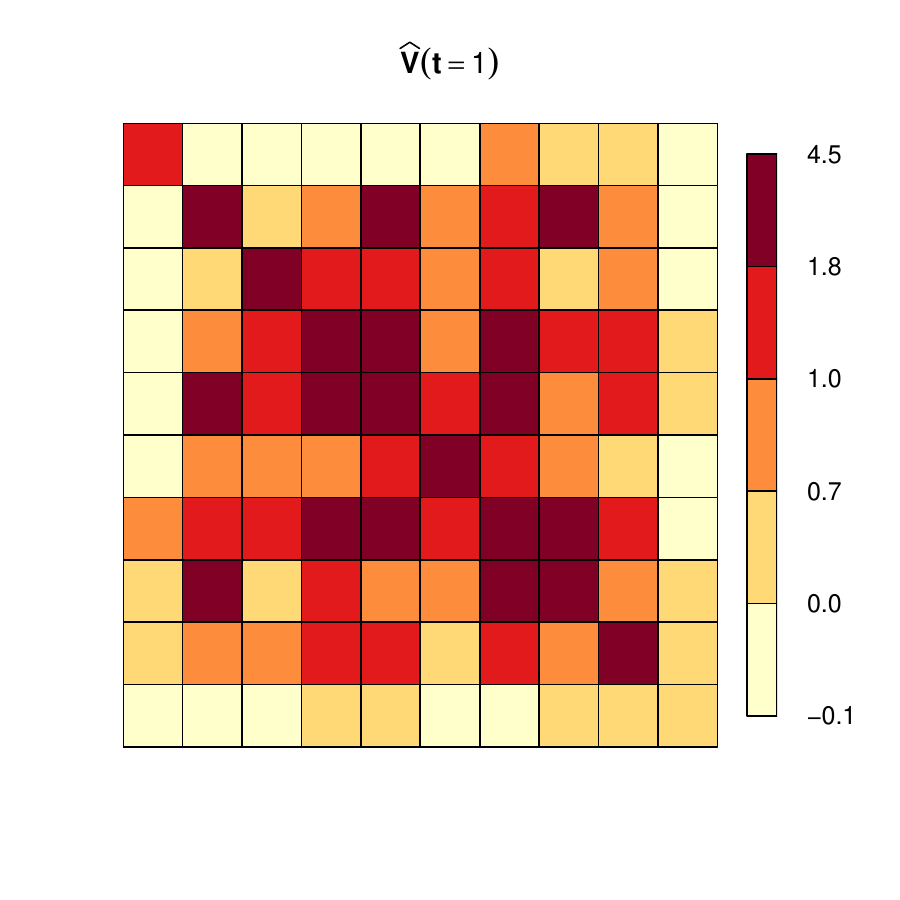} \vspace{-1.5cm}
     \caption{}  
     \label{fg1:vsampleT1}
    \end{subfigure}
    \begin{subfigure}[b]{0.32\textwidth}
     \centering
     \includegraphics[width=\textwidth]{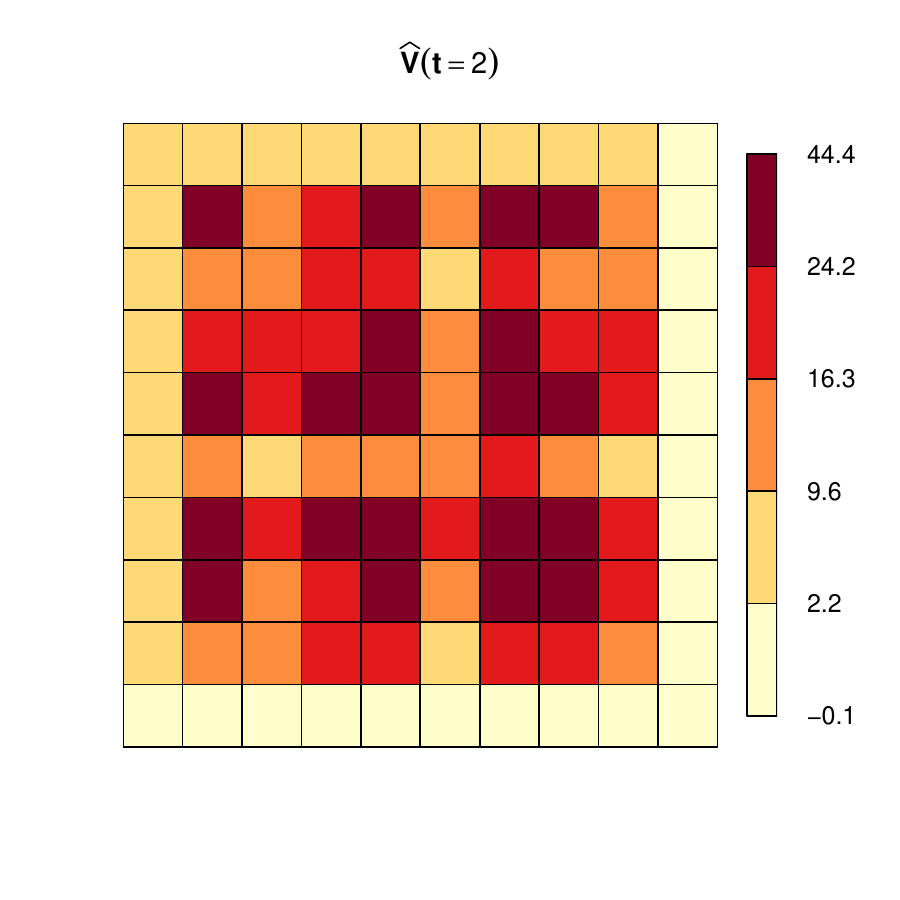} \vspace{-1.5cm}
     \caption{}
     \label{fg:vsampleT2}
    \end{subfigure}
    \begin{subfigure}[b]{0.32\textwidth}
     \centering
     \includegraphics[width=\textwidth]{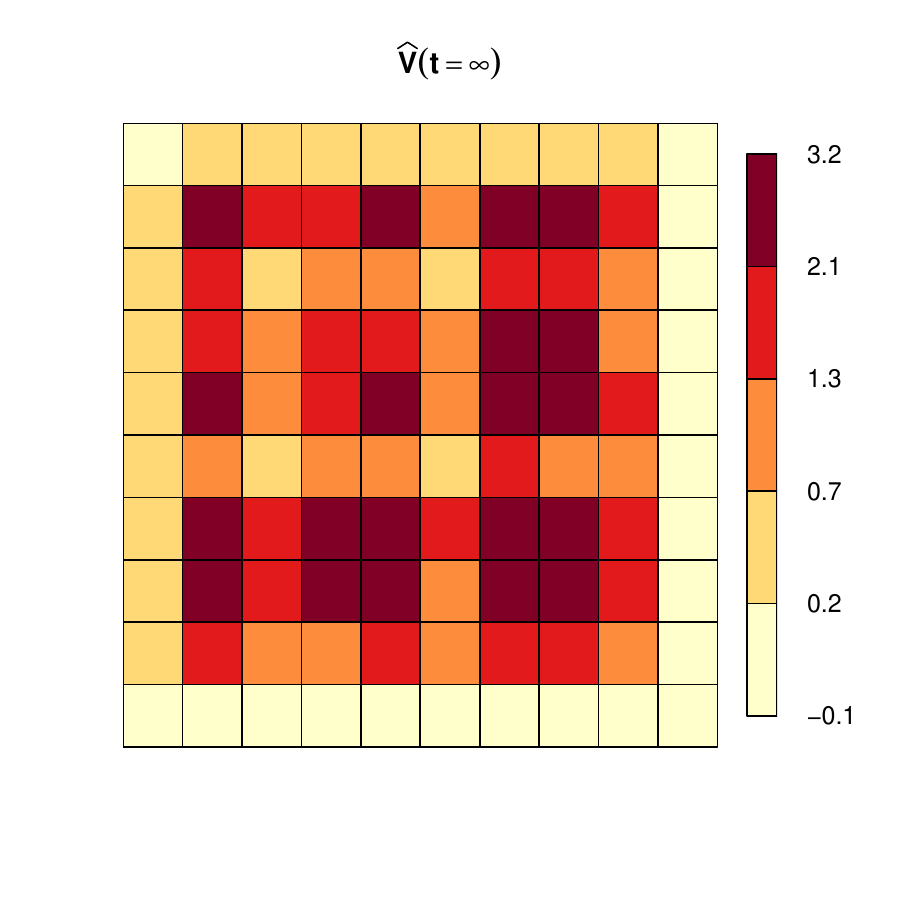} \vspace{-1.5cm}
     \caption{}  
     \label{fg:vsampleTinfty}
    \end{subfigure}
    \caption{ \label{fig:figure1} Visualization for (a) a network with 10 nodes, (b) their adjacency matrix,  
    (c) their geodesic distance matrix. 
    Sample variance-covariance matrix based on 1000 independent replicates under the direct transmission process at (d) $t=1$, (e) $t=2$, and (f) $t = 10000$.}
\end{figure*}

{We introduce one class of transmission processes to illustrate the potential discrepancy between the structure of $\mathbf{A}$ and the dependence structure within one variable, as captured by the variance-covariance matrix of $\bV$}: the direct transmission process. This transmission process provides a general model in which each node's value of a variable is influenced by its adjacent neighbors over ``time'' in a longitudinal way. Here, the time unit should be sufficiently small to capture the influence from one node to other, assuming that influence occurs only during discrete interactions~\citep{ogburn2018challenges, lee2020testing, lee2021network, ogburn2024causal}.   Let $Y^0_i$ denote the outcome value of node $i$ at baseline before initiating any interactions ($i=1,\ldots, n$). Since no interactions have occurred yet, these are independent across $i$. At the subsequent time point ($t=1$), each node has the opportunity to influence its adjacent nodes. Therefore, $Y^{1}_{i}$ may be influenced by the baseline values of its adjacent peers. {As one example}, $Y^1_i$ can be represented as a weighted average of its adjacent peers' baseline values, in addition to its own baseline value: 
\begin{eqnarray*}
Y^1_i = \kappa\sum_{j\neq i}A_{ij}Y^{0}_j + \alpha Y^0_i+ \epsilon_i, 
\end{eqnarray*}
where $\kappa$ and $\alpha$ are scalars, and $\epsilon_i$ is independent random noise. {As this process continues over discrete time points, generating dependencies through multiple paths between nodes, the amount of dependence between two connected nodes $i$ and $j$, e.g., $\text{Corr}(Y^t_i,Y^t_j)$, reflects the cumulative influence from many different paths connecting $i$ and $j$.} {Such a direct transmission process that generates dependencies in the variable through network ties over discrete time points has been considered in various fields of studies, including infectious diseases~\citep{keeling2005networks,wang2018characterizing}, economics~\citep{banerjee2013diffusion,bailey2022peer}, public health~\citep{christakis2007spread,aral2017exercise}, and education~\citep{dishion2011peer,cascio2016first}. 
 
}

Figure~\ref{fig:figure1} demonstrates network distances and the variance-covariance matrices induced by a direct transmission process of multiple $t$ using a toy example of a network with ten nodes. Figures~\ref{fg:adj} and~\ref{fg:geodesic} present their adjacency matrix $\bA$ and pairwise geodesic distance, respectively. While each component in $\bA$ is binary, pairwise geodesic distances are more variable among the connected nodes; only an isolated node (node 10 in Figure~\ref{fg:network}) has an infinite distance to other nodes. 
We generate network dependent vector $\mathbf{Y}^t = (Y^t_{1}, \ldots, Y^t_{10})$ from the direct transmission process where the initial values are $Y^0_i \overset{\text{i.i.d}}{\sim}\mathcal{N}(0,1)$ for $i = 1,\cdots,10$. Figures~\ref{fg1:vsampleT1},~\ref{fg:vsampleT2}, and~\ref{fg:vsampleTinfty} present the sample variance-covariance matrix of $\bY^t$ based on independent 1000 replicates for ${t=1}$, ${t=2}$, and ${t=10000}$ (denoted as $t = \infty$), respectively.  At $t=1$, in addition to adjacent pairs, nodes with a geodesic distance of $2$ (e.g., node $1$ and $7$) are shown to be correlated. At $t=2$ where the influence propagates over two distances (from a node to a node to a node), almost all pairs of connected nodes exhibit dependence. After evolving for a long time (e.g., $t=10000$ possible interactions through network ties), the dependence becomes unstructured, making it challenging to capture it only through $\bA$.\footnote{The scales increase substantially from Figures~\ref{fg1:vsampleT1} to~\ref{fg:vsampleT2} as the value of $\kappa$ is relatively larger, resulting in exponential growth of $\bY^t$, even with only $t=2$. To ensure the process reaches a steady state after a long run ($t = 10000$ in this case) and to prevent the values of $\bY^t$ from growing exponentially over time, we adjusted the parameters of the process, resulting in smaller scales in Figure~\ref{fg:vsampleTinfty} compared to Figure~\ref{fg:vsampleT2}.} These results highlight the structural differences between $\bA$ and $\bV$, even in the simple case.

\subsection{{Temporal dynamics in network and time series data}}

The temporal dynamics that generate network dependence are fundamentally different from the dynamics that generate temporal autocorrelation in time series data. In network settings, 
researchers typically observe one (or at most, a few) snapshots of autocorrelated outcomes for $n$ nodes in a network after the nodes have interacted over a \textit{specific} time period,  say $t=t_{0}$, denoted as $\mathbf{Y}^{t_{0}} = (Y^{t_{0}}_{1}, Y^{t_{0}}_{2}, \ldots, Y^{t_{0}}_{n})$. From now, we refer to this time duration of interactions as the \textit{level of interactions}. Here, network autocorrelation refers to the correlation between two nodes, e.g., $\text{Corr}(Y^{t_{0}}_{i}, Y^{t_{0}}_{j})$. 
In contrast, temporal data often involves the collection of time series $\{ (Y^{t=1}_{i}, Y^{t=2}_{i}, \ldots, Y^{t=T}_{i}): i=1,2,\ldots,n \}$, where the autocorrelation between two time points, e.g., $\text{Corr}(Y^{t}, Y^{t-k})$ refers to temporal autocorrelation, which can be inferred from observations across multiple subjects ($n$).

{When modeling a temporal dependence pattern, \textit{stationarity}---assuming consistent temporal dependence patterns over time---often plays a crucial role, It allows the dependence structure to be captured with relatively few parameters, typically as a function of temporal distance between two time points. For example, under stationarity, $\text{Corr}(Y^t, Y^{t-k})$ depends only on the lag $k$, not on the specific time point $t$. 
In network settings, the dependence between two nodes can become (nearly) independent of the specific number of transmission or interaction steps (e.g., $t$) once that process reaches \textit{equilibrium}. Equilibrium refers to the state in which the dependence structure among nodes stabilizes. At equilibrium, network dependence, like temporal dependence under stationarity, can often be modeled using the network structure (e.g., an adjacency matrix) and a relatively small number of parameters. While equilibrium and stationarity represent different concepts, both allow for simplified modeling of dependence structures through stable, underlying patterns in temporal dynamics. However, unlike time series, determining whether network data has reached equilibrium relies on the researcher’s conjecture or prior knowledge, as there is often insufficient data on network dynamics.}

The implications of approaches for accounting for shared network dependence should critically depend on the assumed nature of transmission process that underlies dependence structures.
We first introduce the pre-whitening approach, which can be applicable when the dependence can be captured through a specific level of interactions (e.g., a specific $t_{0}$). We then introduce the network autocorrelation model, which is more appropriate when the network is assumed to have reached an equilibrium state.

\section{Pre-whitening approaches}\label{sec-premix}

\subsection{Specification of the variance-covariance matrix in pre-whitening approach}

The pre-whitening approach was developed specifically to address shared dependence between variables, which has been dominantly used in temporal autocorrelation~\citep{smith2007comment,bullmore1996statistical,yue2024iterative}. 
Pre-whitening often involves estimating the variance-covariance matrix of $\bV$ using (often many) variables that are believed to share the same dependence structure with $\bX$ and $\bY$. 
Then each of autocorrelated variables $\bX$ and $\bY$ is transformed into $\widehat\bV^{-1/2}\bX:= \bX^*$ and $\widehat\bV^{-1/2}\bY := \bY^*$, respectively. 
This whitening procedure can remove shared dependence in the data by rendering elements each in the pre-whitened variables $\bX^{*}$ and $\bY^{*}$ (possibly asymptotically) statistically independent given that $\bV$ is correctly specified. With $\bX^*$ and $\bY^*$, standard statistical analyses (e.g., fitting a simple linear regression) can be conducted. 

The validity of the pre-whitening approach heavily depends on the specification of $\bV$. 
Researchers across various fields have developed different approaches based on the dependence structure they aim to address, including temporal dependence in the context of time series data in neuroimaging studies (e.g., \citep{lund2006non,luo2020improved}), genetic dependence in the context of GWAS (e.g., \citep{hoffman2013correcting,zhou2012genome}). These approaches may be adapted to network settings if their assumed dependence structures are considered appropriate for modeling network dependence.

Without specifying any transmission process or any level of interactions that result in the dependence structure, one may attempt to estimate $\bV$ directly through the data. However, in a network, estimating $\bV$ of one or a small number of variables is particularly challenging, especially when only one snapshot of autocorrelated variables is observed. A more practical way is to construct $\bV$ through an adjacency matrix $\bA$, which is similar to using the kinship matrix in GWAS.
As non-adjacent pairs have zero elements in $\bA$, it captures only first-order network dependence. Here, ``first-order'' refers to the dependence resulting from paths of length one that connect two nodes (i.e., direct connection). 
The following sections discuss specifications of $\bV$ using higher-order ($>$1) powers of $\bA$ to characterize dependence beyond direct adjacency.

\subsection{Possible specifications under the direct transmission process}\label{sec:pwTonet}

While the relationship between $\bA$ and $\bV$ can be generally very complex, we provide the following Theorem to establish their relationship under the given level of interactions in a very specific direct transmission process. The result can then be used to pre-whiten the variables in network settings.

\begin{theorem}[Variance-covariance matrix in a direct transmission process]\label{thm:theorem1}
   Suppose that 
(i) $\bA$ is symmetric, and  
(ii) $Y^0_i\overset{\text{i.i.d}}{\sim}\mathcal{N}(0,1)$ for $i=1,\ldots, n$. 
Then when $\bY^t = (Y^t_1,\cdots,Y^t_n)$ follows the direct transmission process given by:
       \begin{equation}~\label{directMatrix}
           \bY^t = \kappa\bA\bY^{t-1} + \alpha \bY^{t-1}  + \bepsilon^t, \quad t = 1, 2, \cdots, T,
       \end{equation}
       where $\kappa \in \mathbb{R}$ and $\alpha\in \mathbb{R}$, $\bepsilon^{t} = (\epsilon^t_1,\cdots,\epsilon^t_n), \epsilon^t_i\overset{\text{i.i.d}}{\sim}\mathcal{N}(0,1)$, for $i = 1, \cdots, n$, 
 the variance-covariance matrix $\bV$ of $\mathbf{Y}^{t=1}$ is given by:
 \begin{equation}\label{eq:varTheorem1}
     \bV = \text{Var}(\bY^{t=1}) = \kappa^2\bA^2 + 2\alpha\kappa\bA + (\alpha^2+1)\bI,
 \end{equation}
    where $\bI$ is the $(n \times n)$ identity matrix.
\end{theorem}
While the direct transmission process at $t=1$ in~\eqref{directMatrix} appears to generate autocorrelation only between adjacent pairs, Theorem~\ref{thm:theorem1} demonstrates that $\bV$ is, in fact, related to $\bA$ up to its order two. Consider a network with three nodes, $i,j$ and $l$, where $i$ and $j$ are connected only through $l$, i.e., $A_{il}A_{lj} = 1$ but $A_{ij} = 0$. At $t= 1$ in the direct transmission process, although $i$ and $j$ have not had direct interaction, $Y^{t=1}_i$ and $Y^{t=1}_j$ are statistically correlated through $Y^{t=1}_l$. In this way, the order of $\bA$ would exponentially---not linearly---increase as the number of interactions allowed between adjacent pairs increases. A proof of Theorem~\ref{thm:theorem1} is provided in Appendix~\ref{sec:proof}.

\begin{figure}[!ht]
\centering
    \begin{subfigure}[b]{0.4\linewidth}
     \centering
     \includegraphics[width=\linewidth]{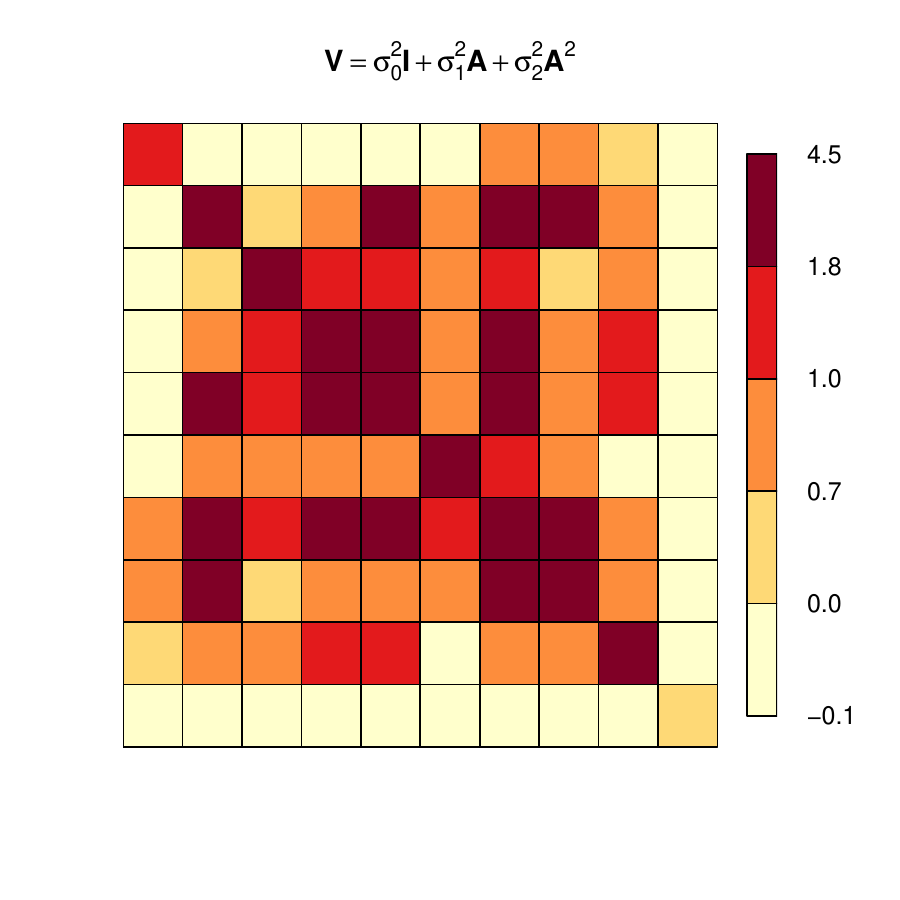}
     \vspace{-1.5cm}
     \caption{}  
    \label{fg:vtrueT1}
    \end{subfigure}
     \begin{subfigure}[b]{0.4\linewidth}
     \centering
     \includegraphics[width=\linewidth]{plots/Sample_t1.pdf}
    \vspace{-1.5cm}
     \caption{}  
     \label{fg:vsampleT1}
    \end{subfigure}
    \caption{ \label{fig:figure2} 
    (a) The variance-covariance matrix $\bV$; (b) the sample variance-covariance matrix under the direct transmission process at $t = 1$.}
\end{figure}

Figure~\ref{fg:vtrueT1} presents $\bV$ in~\eqref{eq:varTheorem1} using the parameters that generate the direct transmission process for the toy example introduced in Figure~\ref{fig:figure1}. Figure~\ref{fg:vsampleT1} presents the sample variance-covariance matrix of $\bY^{t=1}$ based on 1000 independent replicates (same as Figure~\ref{fg1:vsampleT1}), which exhibits a dependence structure very similar to that specified in $\bV$ in Figure~\ref{fg:vtrueT1}. This agreement between theoretical and empirical variance-covariance matrices is consistent with our results in Theorem~\ref{thm:theorem1}.

\begin{remark}
If the adjacency matrix is asymmetric, capturing the full dependence structure may require considering high-order adjacency matrix, its transpose matrix $\bA^\T$, and their product $\bA^{d_1}{\bA^\T}^{d_2}$, where $d_1$ and $d_2$ are scalars. For example, when $\bA$ is asymmetric, the first term in the variance-covariance matrix in Theorem~\ref{thm:theorem1} is given by $\kappa^2\bA\bA^\T$.
\end{remark}

\subsection{General specifications of the variance-covariance structure in a network}

While the results in Theorem~\ref{thm:theorem1} are helpful for understanding the relationship between $\bA$ and $\bV$ within a certain model and a fixed level of interaction ($t_{0} = 1$), they may appear valid only under restrictive conditions. 
Without necessarily assuming a specific form of the transmission process, it is also an option to construct $\bV$ through the higher-order linear combination of $\bA$ given a certain order $d$. For example:
\begin{equation}\label{eq:mixVariance}
    \bV = \sigma^2_0 \bI +  \sigma^2_1\bA + \cdots \sigma^2_d \bA^d,
\end{equation}
where $\sigma^2_0$ represents the variance of independent component, each term $\bA^m $ captures \textit{$m$-order dependence} resulting from length-$m$ paths that connect two nodes, with $\sigma^2_m$ indicating the strength of dependence. 
The value of $\bA^m_{ij} =g$ indicates that there exist $g$ number of length-$m$ paths connecting nodes $i$ and $j$. The variance-covariance matrix defined in~\eqref{eq:mixVariance} assumes that each length-$m$ path connecting nodes $i$ and $j$ contributes to $\text{Cov}(Y_{i}, Y_{j})$ with the same strength of $\sigma^2_m$. However, the impact of $m$-order dependence on $\bV$ may vary across different orders of $m$ ($=1,\ldots,d$). In many cases, it is reasonable to conjecture that dependence induced by nodes farther apart in the path would be negligible~\citep{fowler2010cooperative,papachristos2014network,soininen2007distance,parkinson2018similar}. 

The choice of $d$ is subjective, and researchers may determine the value of $d$ based on their substantive knowledge of how information propagates through network ties and to what extent.  
With the underlying transmission process unknown, empirical evidence on social influence can guide the choice of $d$~\citep{milgram1967small,christakis2009connected,fowler2010cooperative,backstrom2012four}.  
The \textit{six degrees of separation} rule~\citep{milgram1967small,travers1977experimental} suggests that any two individuals in a social network are connected through at most six steps, indicating that setting $d$ to a finite number (e.g., 10) may be sufficient to capture indirect dependence in many networks, regardless of the dynamic of the underlying transmission process. Additionally, the \textit{three degrees of influence} rule~\citep{christakis2009connected,fowler2010cooperative} further states that influence tends to diminish beyond three steps, implying the dependence due to transmission process may be at most significant within $t=3$. In the context of a direct transmission process, these two rules align with each other, as the variance structure in $t=3$ involves adjacency relationships up to $\bA^6$, which exactly corresponds to the six degrees of separation. Furthermore, network studies specifically on social media suggest that the degrees of separation may be even lower ($d=4$)~\citep{backstrom2012four,daraghmi2014we}. This implies if the underlying network structure is denser and the distribution of distances is more right-skewed, a smaller $d$ may be more appropriate, {while the specific transmission processes that result in $\bV$ in~\eqref{eq:mixVariance} for $d >2$ remain unknown to us}.

\subsection{An example of pre-whitening example using linear mixed models}\label{ssec:exlmixed}

In this section, we provide an example of the pre-whitening approach, where the variance-covariance matrix is estimated using a linear mixed model. 
To address the shared genetic dependence, the linear mixed model has been used to represent shared dependence structure by random effects in several studies~\citep{kang2008efficient,kang2010variance,zhou2012genome,sul2018population,lippert2011fast}. More recently, it has also been applied specifically to address the spurious associations due to network dependence~\citep{landon2018patient, matthews2024evaluation}. For example,
the linear mixed model of $\bX$ and $\bY$ can be given by:
\begin{equation}\label{eq:egmix}
    \bY = \beta \bX + \bu + \bepsilon, 
\end{equation}
where $\bu$ represents the random effect associated with network dependence, $\bepsilon\sim \mathcal{N}(0,\sigma^2_0\bI)$ denotes the error term, and $\beta$ is the fixed effect for relationship between $\bX$ and $\bY$. We assume $\bu \sim \mathcal{N}(0,\bV_u)$, where $\bV_u = \sigma^2_1\bA +\cdots + \sigma^2_d\bA^d$ and hence $\bY \sim \mathcal{N}(\beta \bX , \bV)$, where $\bV$ is given in~\eqref{eq:mixVariance}. Although unstructured forms of $\bV_u$ and $\bV$ can be used in \eqref{eq:egmix}, specifying $\bV$ as in \eqref{eq:mixVariance} can substantially reduce the dimension of parameter space, and also can reflect substantive knowledge about transmission dynamics. 
To estimate the unknown parameters, the regression coefficient, and the variance components ($\sigma^2_0,\cdots,\sigma^2_d$), the maximum likelihood estimators can be obtained by maximizing the following log-likelihood function.
\begin{equation}\label{eq:mixLM}
    \begin{split}
          & l(\bY, \bX,\beta,\sigma^2_0,\sigma^2_1, \cdots, \sigma^2_d)  =  -\frac{n}{2}\ln(2\pi)-\frac{1}{2}\ln|\bV| \\&- \frac{1}{2} (\bY -\beta \bX)^\T (\bV)^{-1}(Y-\beta \bX).\\
    \end{split}
\end{equation}
Even numerically maximizing the above log-likelihood function often faces challenges, as it involves computing the inverse of updated $\bV$ in each iteration. The following Theorem simplifies the optimization procedure by spectral decomposition under certain conditions, a similar approach to~\citep{sul2018population} but with a more general structure of $\bV$ as in~\eqref{eq:mixVariance}.
\begin{theorem}[Maximum likelihood estimation with efficient optimization] \label{thm:theorem2}
Assuming $\bA$ is symmetric and $\bY \sim \mathcal{N}(\beta\bX, \bV)$, where $\bV$ is given in~\eqref{eq:mixVariance},
the maximum likelihood estimators for $(\beta,\sigma^2_0,\sigma^2_1,\cdots,\sigma^2_d)$ in~\eqref{eq:egmix}
are obtained by maximizing the following log-likelihood function:
    \begin{equation*}\label{eq:simmixLM}
    \begin{split}
            &l(\bY, \bX,\beta,\sigma^2_0,\sigma^2_1, \cdots, \sigma^2_d)   = -\frac{n}{2}\ln(2\pi) \\&-\frac{1}{2}\sum^n_{i=1} \ln(\sigma^2_0 + \sigma^2_1 \lambda_i + \cdots +  \sigma^2_d \lambda^d_i) \\ & - \frac{1}{2} \bZ^\T(\sigma^2_0 \bI + \sigma^2_1\bD + \cdots + \sigma^2_d \bD^d)^{-1}\bZ,
    \end{split}
    \end{equation*}
    where $\bZ = \bU^\T(\bY-\beta\bX)$, $\bA = \bU\bD\bU^\T$, $\bU$ is the orthogonal matrix of eigenvectors of $\bA$ and \\
    $\bD =\text{diag}[\lambda_1,\lambda_2,\cdots,\lambda_n]$ is a diagonal matrix of eigenvalues of $\bA$.
\end{theorem}
A proof of Theorem~\ref{thm:theorem2} is provided in Appendix~\ref{sec:proof}. The estimated coefficient $\hat\beta$ represents the relationship of $\bX$ and $\bY$ accounting for the shared network dependence. If we assume that $\bV$ is the variance-covariance matrix for both $\bY$ and $\bX$, then $\hat\beta$ is equivalent to the estimator obtained from linear regression of the whitened variables, $\bX^* = \widehat\bV^{-1/2}\bX $ and $\bY^* =\widehat\bV^{-1/2}\bY$.

\section{Network Autocorrelation Model}\label{sec-nam}

\subsection{Model specification and estimation}\label{ssec-namspecify}
The network autocorrelation model has been widely used to address or account for autocorrelation within a variable of interest in a network~\citep{neuman2010structure,leenders2002modeling,fujimoto2011network,dow2007galton,doreian1984network,matthews2024evaluation,sewell2017network}. 
The standard form of the network autocorrelation model is given by:
\begin{eqnarray}\label{eq:nam}
   \bY  = \rho \bW \bY + \beta_\text{NAM} \bX + \bepsilon, 
\end{eqnarray}
where $\bepsilon \sim \mathcal{N}(0, \sigma^2\bI)$ denotes the error terms.

As an extension of standard linear regression model, the network autocorrelation model relaxes the assumption of independence in $\bY$ by incorporating the term $\rho\bW\bY$, allowing $\bY$ to exhibit network dependence. The coefficient $\beta_\text{NAM}$ represents the linear relationship between $\bY$ and $\bX$ after addressing the autocorrelation only within $\bY$, where $\bX$ is not necessarily network-dependent. In this sense, the network autocorrelation model may be appropriate for our target inference to examine the relationship between $X_i$ and $Y_i$, with the conditions for its applicability detailed in the next section.

Each element of the weight matrix $\bW$, $W_{ij}$, represents the extent to which $Y_i$ depends on $Y_j$ while the scalar parameter $\rho$  quantifies the strength of the network autocorrelation. The selection of $\bW$ is crucial, as it defines the network influence pattern and directly impacts parameter estimation and inference. Researchers typically determine $\bW$ based on their prior knowledge and available network information (e.g., $\bA$).
One common specification is the adjacency matrix $\bA$ (e.g.,~\cite{sewell2017network,fujimoto2011network}), and its row-normalized version of $W_{ij} = A_{ij}/\sum_j A_{ij}$ is also commonly used~(e.g., \cite{li2021comparison,dittrich2020network,dittrich2017bayesian,mcpherson2005environmental,gimpel2003political}).
The parameter space for $\rho$ also depends on the choice of $\bW$.
Rearranging~\eqref{eq:nam}, we have the following representation: $\bY = (\bI - \rho \bW)^{-1} (\beta_\text{NAM}\bX + \bepsilon)$, given that $(\bI - \rho \bW)$ is non-singular. 
A detailed discussion on the selection of $\bW$ and the feasible range of $\rho$ can be found in~\cite{leenders2002modeling} and \cite{ver2018relationship}, respectively.

Several methods have been developed to estimate parameters in the network autocorrelation model. \cite{doreian1981estimating} proposes likelihood function to estimate parameters $\rho$ and $\beta_\text{NAM}$, which has been predominantly used in
subsequent studies~\citep{doreian1984network,mizruchi2008effect, wang2014statistical}. Additionally, \cite{kelejian1998generalized} propose a two-stage least squares method through an instrument variable $\bZ:= (\bX, \bW\bY)$, involving regress $\bX$ on $\bZ$ in the first stage, and regress $\bY$ on fitted values from the first stage as the second stage.
\cite{dittrich2017bayesian} provide a Bayesian estimation approach where an informative prior for $\rho$ is recommended to use. A comprehensive review of different estimation methods for the network autocorrelation model can be found in~\cite{li2021comparison}. We present the likelihood-based estimation, where  $\widehat\beta_\text{NAM}$, $\hat\rho$, and $\hat\sigma^2$ are derived from maximizing the following log-likelihood:
\begin{equation}\label{mle for nam}
\begin{split}
    & l(\bY,\bX,\rho, \beta_\text{NAM},\sigma^2) \\ &=  -\frac{n}{2}\ln(2\pi\sigma^2)-\frac{n}{2}\ln(\det(\bI-\rho \bW))\\&-\frac{1}{2\sigma^2}((\bI - \rho\bW)\bY  - \beta_\text{NAM}\bX)^\T((\bI - \rho\bW)\bY - \beta_\text{NAM}\bX) 
\end{split}
\end{equation}
Maximizing the above log-likelihood function typically requires less computational burden to optimize~\eqref{mle for nam}, compared to optimizing~\eqref{eq:mixLM} in Section~\ref{sec:pwTonet}, as it avoids computing the matrix inverse in each iteration.

\subsection{Underlying assumptions}

We have observed the possibility of directly using the network autocorrelation model to infer the statistical association between two network-dependent variables, even without assuming autocorrelation in the exposure variable $\bX$.
Here, we elaborate on the assumptions required to use this model, particularly those related to the transmission process it implies. 
Specifically, similar to our discussion of the pre-whitening approach in Section~\ref{sec:pwTonet}, we demonstrate the connection between the direct transmission and the dependence structure implied by the model. For simplicity, consider using $\bA$ in place of $\bW$ and assume the following direct transmission process:
\begin{equation*}~\label{direct:nam}
    \bY^{t} = \rho\bA\bY^{t-1} + \beta_\text{NAM}\bX + \bepsilon, \quad t=1,2,\ldots, T
\end{equation*}
with $\bY^0 \sim \mathcal{N}(0,\bI)$. Then as the process evolves, we have the following form at $t$:
\begin{equation*}
    \bY^{t} =   \rho^t \bA^t\bY^{0}  +  \sum^{t-1}_{q=0}(\rho^q\bA^q)(\beta_{\text{NAM}}\bX + \bepsilon).
\end{equation*}
Under certain conditions on the eigenvalues of $\bA$, $\rho^t\bA^t$ vanishes as $t \to \infty$~\citep{golub2013matrix}. Therefore, $\lim_{t\to \infty}\rho^t \bA^t\bY^{0} = 0$. Then we have:
\begin{equation*}
    \lim_{t\to \infty}\bY^{t} = (\bI - \rho\bA)^{-1} (\beta_\text{NAM}\bX + \bepsilon),
\end{equation*}
where $ (\bI - \rho\bA)^{-1} = \sum^{\infty}_{q=0}(\rho^{q}\bA^q) $ by the Neumann series, assuming $(\bI - \rho\bA)$ is non-singular~\citep{lesage2009introduction}. With $\bY := \lim_{t\to\infty} \bY^{t}$ representing the outcome after a "long-run" process or an infinite level of interactions, we can interpret $\bY$ in the network autocorrelation model~\eqref{eq:nam} as the outcome at the equilibrium state reached following the long-term direct transmission process through network ties specified in $\bA$. 

The network autocorrelation model also dictates the explicit form of the variance-covariance matrix of $\bY$ in~\eqref{eq:nam}.
Given the Neumann series,
\begin{equation}~\label{namVar}
\begin{split}
     \text{Var}(\bY)& = (\bI - \rho\bA)^{-1}\text{Var}(\beta_\text{NAM}\bX + \bepsilon)(\bI - \rho\bA)^{-1}\\& =\sigma^2\sum^\infty_{q=0}(\rho\bA)^q.
\end{split}
\end{equation}
For this reason, the network autocorrelation model can also be interpreted as a special form of pre-whitening applied only to the outcome variable $\bY$ with the variance-covariance matrix of $\bV = \sigma^2\sum^\infty_{q=0}(\rho\bA)^q$. 

Even though the network autocorrelation model uses only the first-order $\bA$ as an input (as the weight matrix $\bW$), the variance-covariance matrix induced by the model actually depends on all higher-order terms of $\bA^q$, for $q=0,1,\ldots, \infty$. In other words, the model assumes non-zero autocorrelation between the two nodes as long as they are connected. Assuming $\rho$ is less than $1$, the variance-covariance matrix defined in~\eqref{namVar} enforces a decay pattern for each order of dependence. Particularly, the value $\sigma^2\rho^q$ decays as $q$ increases, indicating that dependence induced by nodes farther apart would diminish, as 
$\lim_{q\to\infty}\sigma^2\rho^q = 0$. The strength of high-order dependence (e.g., $\sigma^2\rho^2$) must be weaker than the strength of the low-order dependence (e.g., $\sigma^2\rho$). This decay pattern is similar to the autocorrelation structures typically assumed in temporal and spatial models.
Consequently, while the model allows a parsimonious representation of $\bV$ via $\bA$, it is not flexible enough to accommodate settings where high-order dependence is strong or equal to lower-order dependence.

\section{Simulation}\label{sec-sim}

The simulation study aims to evaluate the effectiveness of the aforementioned two approaches in addressing spurious associations due to network dependence.  We generate network-dependent data under different transmission processes and examine the performance of the two approaches.

\subsection{Simulation settings}

We first generate random networks from the Erdős–Rényi model~\citep{erdds1959random}, where each network tie is formed with a fixed probability, independently of other ties. We set the sample size to $n = 500$ nodes and fix the number of ties at $500$.

Given the network structure, we generate $\bY = (Y_{1}, \ldots, Y_{n})$ through the two transmission processes to generate network dependence: (a) the direct transmission process at $t=1$, and (b) the equilibrium transmission process. For each process, we vary the strength of network dependence as weak, medium or strong.
We \textit{independently} generate $\bX$ in the same way under each network dependence setting. 
For each simulation, with network-dependent $\bX$ and $\bY$, we measure a statistical association between $X_{i}$ and $Y_{i}$ using $n$ observations and examine whether there is significant evidence of inflated type-I errors based on the results of 500 replicates. 
We consider three association measures here: (i) the coefficient $\beta_\text{OLS}$ from a simple linear model, and (ii) the distance correlation $dCorr$, which captures the relationships beyond linearity~\citep{szekely2007measuring}, and (iii) linear coefficient $\beta_{\text{NAM}}$ obtained directly from the network autocorrelation model~\eqref{eq:nam}. 
More details about data generating models for the simulation studies can be found in Appendix~\ref{sec:simsup}.

We apply the two approaches to network-dependent ($\bX, \bY$): (1) pre-whitening and (2) network autocorrelation model. To apply (1), we specify the variance-covariance matrix for both $\bX$ and $\bY$ as follows.
\begin{equation*} \label{eq:second}
    \bV = \sigma^2_0 \bI + \sigma^2_1\bA + \sigma^2_2\bA^2,
\end{equation*}
which was proven to be the correct form of the variance-covariance matrix for the direct transmission process at $t=1$, according to Theorem~\ref{thm:theorem1} in Section~\ref{sec:pwTonet}.
Maximum likelihood estimation is then applied to each of $\bX$ and $\bY$ separately to derive the estimated variances $\widehat\bV_X$ and $\widehat\bV_Y$.
The pre-whitening then transforms ($\bX$,$\bY$) into network-independent versions ($\bX^*$,$\bY^*$) using $\bX^* = \widehat\bV^{-1/2}_X\bX$ and $\bY^* = \widehat\bV^{-1/2}_Y\bY$, respectively, each of which is used later to obtain association measures.  

\begin{table*}[!ht]
    \centering
    \caption{Rejection rates for testing the null hypothesis of no statistical correlation between $X_{i}$ and $Y_{i}$, using different association measures for $(\bX, \bY)$, $(\bX^*, \bY^*)$, and $(\bX^*_{\text{NAM}}, \bY^{*}_{\text{NAM}})$ across (a) the direct transmission ($t= 1$) and (b) the equilibrium transmission processes, under three different strengths of network dependence for each process.}
    \label{tab:rates}
    \resizebox{0.8\textwidth}{!}{
   \begin{tabular}{ll|ccc|cc|cc}
    \toprule
        &  & \multicolumn{3}{|c|}{$(\bX,\bY)$} &  \multicolumn{2}{|c|}{$(\bX^*,\bY^*)$} & \multicolumn{2}{c}{$(\bX^*_{\text{NAM}},\bY^*_{\text{NAM}})$}  \\
       &   &  $\beta_{\text{OLS}}$ &  $\text{dCorr}$ &$\beta_{\text{NAM}}$&  $\beta_{\text{OLS}}$ &  $\text{dCorr}$ & $\beta_{\text{OLS}}$&  $\text{dCorr}$  \\
     \hline
     \multicolumn{8}{l}{(a) Direct transmission process at $t =1 $}\\
     \hline
       Strength& Weak &0.222 & 0.488 & 0.178 & 0.058 & 0.086 & 0.220 & 0.434\\ 
     &         Medium &0.194 & 0.624 & 0.194 & 0.058 & 0.064 & 0.212 & 0.610\\ 
     &         Strong &0.192 & 0.844 & 0.172 & 0.030 & 0.042 & 0.188 & 0.842\\ 
     \hline
   \multicolumn{8}{l}{(b) Equilibrium transmission process}    \\
     \hline
     Strength  & Weak & 0.302 & 0.294 & 0.060 & 0.068 & 0.088 & 0.052 & 0.058\\
             & Medium & 0.452 & 0.528 & 0.057 & 0.236 & 0.204 & 0.042 & 0.056\\
             & Strong & 0.878 & 0.956 & 0.060 & 0.684 & 0.652 & 0.044 & 0.066\\
     \bottomrule
    \end{tabular} }
 
\end{table*}

We fit the network autocorrelation model with $(\bX, \bY)$ for two purposes; first, to obtain the association measure $\beta_{\text{NAM}}$ directly from the model, or second, to obtain the ``pre-whitened'' $(\bX^*_\text{NAM}, \bY^*_\text{NAM})$ through the variance estimated in the model. We fit the following model for the first purpose: 
\begin{equation} \label{eq:nam_sim}
    \bY = \rho \bA \bY + \beta_{\text{NAM}}\bX + \bepsilon, \quad \bepsilon \sim \mathcal{N}(0,\sigma^2\bI).
\end{equation}
For the first purpose, we use the estimated coefficient $\widehat\beta_{\text{NAM}}$ to evaluate whether the network autocorrelation model can mitigate the issue of spurious associations. 
On the other hand, as the network autocorrelation model specifies a particular variance-covariance structure for $\bY$, given by $ \bV_{Y,\text{NAM}} = \sigma^2(\bI - \rho\bA)^{-1}(\bI - \rho\bA)^{-1}$, we can construct the estimated variance-covariance matrix $\widehat\bV_{Y,\text{NAM}}$ by plugging in the estimated $\hat\rho$ and $\hat\sigma^2$ from the network autocorrelation model without any covariates. Similarly, assuming the same dependence structure between $\bX$ and $\bY$, we obtain $\widehat\bV_{X,\text{NAM}}$ by using $\bX$ as the response variable in~\eqref{eq:nam_sim}, also without any covariates. 
We then pre-whiten each variable by transforming $\bX$ and $\bY$ into $\bX^*_{\text{NAM}} =  \widehat\bV_{X,\text{NAM}}^{-1/2}\bX$ and $\bY^*_{\text{NAM}} =  \widehat\bV_{Y,\text{NAM}}^{-1/2}\bY$.

We evaluate the performance of each approach by (i) and (ii) for all network-dependent $(\bX,\bY)$, as well as the pre-whitened variables, $(\bX^*,\bY^*)$ and $(\bX^*_{\text{NAM}},\bY^*_{\text{NAM}})$. 
For network-dependent $(\bX, \bY)$, we apply simple linear regression, calculate the distance correlation, and fit the network autocorrelation model~\eqref{eq:nam_sim}. For the pre-whitened variables $(\bX^*,\bY^*)$ and $(\bX^*_\text{NAM},\bY^*_\text{NAM})$, we apply simple linear regression and calculate the distance correlation.

\subsection{Simulation Results}\label{sec-sim.re}

Table~\ref{tab:rates} presents the rejection rates of the null hypothesis of no association between $X_{i}$ and $Y_{i}$ across different transmission processes (indicated by each row) using different approaches and association measures (indicated by each column).

Without pre-whitening, the tests using network-dependent $(\bX,\bY)$ produce significant type-I errors, except in the case of $\beta_{\text{NAM}}$ under (b) the equilibrium transmission process where the dependence structure in $\bY$ is correctly specified. 
This result suggests that the network autocorrelation model alone could be sufficient to address the spurious associations if the dependence structure in the outcome variable is correctly specified (e.g., equilibrium state). 
However, the rejection rates of tests on $\beta_\text{NAM}$ are slightly higher than $0.05$, also suggesting that the network autocorrelation model, at most, partially accounts for spurious associations when both the response and covariate variables are network-dependent. These results align with recent literature~\citep{matthews2024evaluation}, which reports that while the network autocorrelation model can mitigate false positive rates, it does not achieve a desirable level (i.e., below 0.05).   
In general, the measure $\text{dCorr}$ yields higher rejection rates than $\beta_{\text{OLS}}$, as distance correlation captures both linear and nonlinear relationships, whereas the linear coefficient reflects only linear associations.

With the pre-whitened $(\bX^*, \bY^*)$, the rejection rates using $\beta_{\text{OLS}}$ and $\text{dCorr}$ remain around the nominal value of 0.05 under (a) the direct transmission process at $t=1$. This result validates our findings in Theorem~\ref{thm:theorem1}. However, under (b) the equilibrium transmission process, the pre-whitening $\bV$ is misspecified for the local dependence, leading to inflated type-I errors, with rejection rates increasing as the strength of network dependence grows. 
With the pre-whitened $(\bX^*_{\text{NAM}}, \bY^*_{\text{NAM}})$ by the $\bV$ specified in the network autocorrelation model~\eqref{eq:nam_sim} without any covariates, the rejection rates are close to the nominal value, slightly above 0.05 under (b) the equilibrium transmission process, while they are substantially higher under (a) the direct transmission process at $t=1$. The results demonstrate that the network autocorrelation model can be effectively used to obtain the pre-whitened variables, which can then be applied in separate models or measures to estimate the association between the variables. However, it may be challenging to compute the inverse matrix $(\bI - \hat\rho\bA)^{-1}$ to obtain the pre-whitened variables $\bX^{*}_{\text{NAM}}$ and $\bY^{*}_{\text{NAM}}$.

\section{Discussion}
\label{sec-dis}

In human subjects research, where a variable of interest represents behaviors or health outcomes, there is typically no inherent structure governing autocorrelation among individuals. {Unlike other dependence settings, researchers usually must rely on prior knowledge to determine the level of interaction (e.g., the choice of $t_{0}$) that result in autocorrelation, especially when only a single snapshot of the data is available. If the observed autocorrelated outcomes in the network are believed to result from a relatively small number of interactions (e.g., a smaller value of $t_0$), then the variance-covariance matrix using higher-order terms of the adjacency matrix, as in \eqref{eq:mixVariance} can be applied up to a moderate degree (e.g., $d = 10$) to pre-whiten the network-dependent variables.   In contrast, if the autocorrelated outcomes are believed to remain relatively stable even as more interactions are allowed (i.e., the process is believed to have reached equilibrium), a network autocorrelation model and related recent methodologies can be used directly for regression or to obtain the variance-covariance matrix from the model.}

{Without specifying any level of interactions leading to autocorrelation, one may attempt to estimate the variance-covariance matrix $\bV$ directly from the data.} However, estimating the variance-covariance matrix $\bV$ of one or a small number of variables is particularly challenging in a network.
While one might infer the dependence structure through multiple covariates under the assumption that they share similar patterns of dependence as in the case of genetic autocorrelation, the number of such covariates available is often quite limited. 
In such cases, researchers analyzing statistical associations between two variables sampled from interconnected subjects can have two directions: (1) explicitly constructing $\bV$ based on available network information (e.g., $\bA$) to pre-whiten each variable and then measuring their association, or (2) incorporating the structure within a regression model to jointly estimate the variance and regression coefficients. For the latter approach, we find that both linear mixed models used in GWAS and network autocorrelation models can help reduce spurious associations, provided that $\bV$ is correctly specified. However, the assumptions underlying each method regarding the transmission process that generates network dependence can be highly restrictive in practice. Moreover, the effectiveness of reducing spurious associations is sensitive to the specific assumptions implied in the transmission process.

Recently, there has been a growing number of attempts to model specific forms of dependence using Markov properties, allowing the joint density of $(\bY, \bX)$ to be factorized into a product of several conditionally independent components. 
In this framework, each variable (e.g., $\bX$ or $\bY$) is treated as a single realization of a particular type of Markov random field over the network, similar to certain equilibrium distributions~\citep{ogburn2020causal, tchetgen2021auto, ogburn2024causal, wu2024network, menzel2025fixed}. While this approach offers statistical (and also causal) advantages, it can also only capture highly specific dependence structures that may not hold in practice. 

Our study opens up potential future directions. First, it is common for $\bX$ and $\bY$ to share the same network dependence structure, though they may only partially overlap. In such cases, the extent of spurious associations is expected to decrease, and any remaining biases could be addressed more efficiently and effectively compared to the case with full overlaps.
Second, similar to GWAS, if multiple variables are available, their sample variances could be utilized to obtain insights into $\bY$. In network studies, an important avenue for exploration is how to leverage the typically small set of covariates, in addition to $\bA$, to better learn about shared dependence structures.  

\section*{Acknowledgement}
This research was supported by Grant Number 5P20GM103645 from the National Institute of General Medical Sciences.

\bibliographystyle{Chicago} 
\bibliography{sample}

\newpage
\begin{appendix}
\section*{\Large{Appendix}}

\end{appendix}
\begin{appendix}
\section{Proofs}
\label{sec:proof}
\begin{proof}[Proof of Theorem~\ref{thm:theorem1}]
The $\bY^1$ is given by: 
\begin{align*}
    \bY^1 = \kappa\bA\bY^0 + \alpha\bY^0 + \bepsilon^1,
\end{align*}
The variance matrix is:
\begin{align*}
    \text{Var}(\bY^1) = \text{Var}((\kappa\bA + \alpha\bI)\bY^0 + \bepsilon^1)
\end{align*}
Expanding it to have:
\begin{align*}
     \text{Var}(\bY^1) 
      & = (\kappa\bA + \alpha\bI)\text{Var}(\bY^0) (\kappa\bA + \alpha\bI)^\T + \text{Var}(\epsilon^1)
\end{align*}
Given $\bY^0 \sim \mathcal{N}(0,\bI)$ and $\bepsilon^1 \sim \mathcal{N}(0,\bI)$, we get:
\begin{align*}
   \text{Var}(\bY^1) = (\kappa\bA + \alpha\bI)(\kappa\bA + \alpha\bI)^\T + \bI
\end{align*}
Since $\bA$ is symmetric, we have:
\begin{align*}
    \text{Var}(\bY^1) =  \kappa^2\bA^2 + 2\alpha\kappa\bA + (\alpha^2 + 1)\bI,
\end{align*}
which completes the proof.
\end{proof}

\begin{proof}[Proof of Theorem~\ref{thm:theorem2}]
Given the eigendecomposition of $\bA$:$\bU\bD \bU^\T$, the high-order $\bA^m$ can be written as:
\begin{align*}
    \bA^{m} &= (\bU\bD\bU^\T)^m = \bU\bD^m\bU^{\T}
\end{align*}
since $\bU^\T\bU = \bI$. Therefore, we can write the variance structure in~\eqref{eq:mixVariance} as:
\begin{equation}\label{proof:V}
    \begin{split}
        \bV &= \sigma^2_0\bU\bI\bU^\T + \sigma^2_1\bU\bD\bU^\T +  \cdots +\sigma^2_d\bU\bD^d\bU^\T\\
        & =  \bU(\sigma^2_0 \bI + \sigma^2_1 \bD +  \cdots + \sigma^2_d\bD^d)\bU^\T
    \end{split}
\end{equation}
where $\bD$ is diagonal with eigenvalues $(\lambda_1,\lambda_2,\cdots,\lambda_n)$, the high-order $\bD^d$ is diagonal with $(\lambda^d_1,\lambda^d_2,\cdots,\lambda^d_n)$. 
The inverse of $\bV$ is:
\begin{equation}\label{proof:inverseV}
    \begin{split}
        \bV^{-1} = \bU(\sigma^2_0 \bI + \sigma^2_1 \bD +  + \cdots + \sigma^2_d\bD^d)^{-1} \bU^{\T},
    \end{split}
\end{equation}
The matrix $(\sigma^2_0 \bI + \sigma^2_1 \bD +  \cdots + \sigma^2_d\bD^d)$ is diagonal, and its inverse can be computed by taking the reciprocal of elements along the diagonal.

Next, we compute the determinate $|\bV|$ by ~\eqref{proof:V}:
\begin{equation*}
    \begin{split}
        |\bV| & = |\bU||\sigma^2_0 \bI + \sigma^2_1 \bD +  \cdots + \sigma^2_d\bD^d| |\bU^\T|\\
        & = |\sigma^2_0 \bI + \sigma^2_1 \bD +  \cdots + \sigma^2_d\bD^d|\\
        & = \prod^n_{i=1}(\sigma^2_0 + \sigma^2_1\lambda_i + \cdots + \sigma^2_d\lambda^d_i)
    \end{split}
\end{equation*}
since $\bU$ is an orthogonal matrix and $|\bU| = |\bU^\T| = 1$. Taking logarithm, we get:
\begin{equation}\label{proof:lnV}
    \begin{split}
         \ln|\bV| = \sum^n_{i=1}\ln(\sigma^2_0 + \sigma^2_1\lambda_i + \cdots+ \sigma^2_d\lambda^d_i)
    \end{split}
\end{equation}

Plugging \eqref{proof:lnV} and \eqref{proof:inverseV} into \eqref{eq:mixLM}, we have:
\begin{equation*}
    \begin{split}
         &l(\bY, \bX,\beta,\sigma^2_0,\sigma^2_1, \cdots, \sigma^2_d)   = -\frac{n}{2}\ln(2\pi) \\&-\frac{1}{2}\sum^n_{i=1} \ln(\sigma^2_0 + \sigma^2_1 \lambda_i + \cdots +  \sigma^2_d \lambda^d_i) \\ & - \frac{1}{2} \bZ^\T(\sigma^2_0 \bI + \sigma^2_1\bD + \cdots + \sigma^2_d \bD^d)^{-1}\bZ,
    \end{split}
\end{equation*}
which completes the proof.
\end{proof}

\section{Additional Simulation Studies}

\label{sec:simsup}
We set the sample size to $n = 500$ nodes with the number of ties at $500$. 
For (a) the direct transmission process at $t = 1$, initial values $\bY^0 = (Y^0_{1},\cdots, Y^0_n)$ are generated with $Y^0_i \overset{i.i.d}{\sim} \mathcal{N}(0,1)$ for $i = 1,\cdots,n$. The network-dependent vector $\bY = \bY^{1} = (Y^{1}_1,\cdots,Y^1_n)$ is generated by:
\begin{equation*}
    \begin{split}
        \bY^1 = \kappa \bA \bY^0 + \alpha \bY^0 + \bepsilon,  \bepsilon = (\epsilon_1,\cdots,\epsilon_n),  \epsilon_i\overset{\text{i.i.d.}}{\sim}\mathcal{N}(0, 0.1^2)
    \end{split}
\end{equation*}
We consider three specifications of $(\kappa,\alpha)$, where $(0.7,0.3)$, $(0.8,0.2)$, and $(0.9,0.1)$ for the strength of network dependence as weak, medium, and strong, respectively. We independently generated $\bX$ in the same way under each network dependence setting.

For (b) equilibrium transmission process, we generate the equilibrium state of the network dependence data by the network autocorrelation model:
\begin{equation*}
    \begin{split}
        \bY &= (\bI - \rho\bA)^{-1}(\bY^0 + \bepsilon), \bepsilon = (\epsilon_1, \cdots,\epsilon_n), \epsilon_i \overset{\text{i.i.d.}}{\sim} \mathcal{N} (0,0.1^2)
    \end{split}
\end{equation*}
where initial values $\bY^0 = (Y^0_{1},\cdots, Y^0_n)$ are generated with $Y^0_i \overset{i.i.d}{\sim} \mathcal{N}(0,1)$ for $i = 1,\cdots,n$. We set $\rho = (0.25, 0.27, 0.29)$ to represent the weak, medium, and strong network dependence, respectively. In the same way and under each network dependence setting, we independently generate $\bX$.
\bigskip

\end{appendix}

\end{document}